\documentclass{IEEEtran}

\usepackage[inkscapeformat=eps]{svg}
\usepackage{amsmath, amssymb, amsfonts}
\usepackage{graphicx}

\usepackage{amsthm}
\usepackage{underscore}
\usepackage[noEnd=true]{algpseudocodex}
\usepackage{algorithm}
\usepackage{enumerate}

\usepackage{gould_macros}

\theoremstyle{definition}
\newtheorem{theorem}{Theorem}
\newtheorem{proposition}{Proposition}
\newtheorem{lemma}{Lemma}

\newtheorem{definition}{Definition}
\newtheorem{assumption}{Assumption}

\newtheorem{remark}{Remark}

\newcommand{\numrows}{m}
\newcommand{\numcols}{n}
\newcommand{\rowspace}{\Delta(\numrows)}
\newcommand{\colspace}{\Delta(\numcols)}
\newcommand{\Gspace}{\R^{\numrows \times \numcols}}

\newcommand{\rowNE}{x_{\mathrm{NE}}}
\newcommand{\colNE}{y_{\mathrm{NE}}}
\newcommand{\rowSP}{x_{\mathrm{SP}}}
\newcommand{\colSP}{y_{\mathrm{SP}}}
\newcommand{\rowSV}{x_{\mathrm{SV}}}
\newcommand{\colSV}{y_{\mathrm{SV}}}

\newcommand{\Gpert}{{G'}}
\newcommand{\Dmag}{\Delta}
\newcommand{\Dspace}{\mathcal{D}}

\newcommand{\rowRobust}{x_{\mathrm{r}}}
\newcommand{\deceptionRobust}{D_{\mathrm{r}}}
\newcommand{\colRobust}{y_{\mathrm{r}}}

\newcommand{\rowIntelligent}{x_{\mathrm{i}}}
\newcommand{\deceptionIntelligent}{D_{\mathrm{i}}}
\newcommand{\colIntelligent}{y_{\mathrm{i}}}

\newcommand{\rowIntelligentFeas}{\overline{\rowIntelligent}}
\newcommand{\deceptionIntelligentFeas}{\overline{\deceptionIntelligent}}
\newcommand{\colIntelligentFeas}{\overline{\colIntelligent}}

\newcommand{\victimStrategies}[1]{\Phi(#1)}
\newcommand{\victimStrategiesRobust}[1]{\Phi_{r}(#1)}
\newcommand{\victimStrategiesSubrational}[2]{\Phi_{#2}(#1)}

\newcommand{\vPerceived}{v_{\mathrm{p}}}
\newcommand{\vDualPerceived}{u_{\mathrm{p}}}
\newcommand{\victimDualAction}{\omega}

\newcommand{\vIndMax}{v^{\star}}
\newcommand{\Vind}[1]{\mathcal{V}(#1)}
\newcommand{\vmin}{v_{\mathrm{min}}}
\newcommand{\vmax}{v_{\mathrm{max}}}
\newcommand{\vcurr}{v_{\mathrm{curr}}}
\newcommand{\vbest}{v_{\mathrm{best}}}
\newcommand{\Dcurr}{D_{\mathrm{curr}}}
\newcommand{\colCurr}{y_{\mathrm{curr}}}

\newcommand{\Dlim}{\hat{D}}
\newcommand{\ylim}{\hat{y}}

\newcommand{\binSearchTol}{\delta}
\newcommand{\binSearchOutSeq}{S}

\def\BibTeX{{\rm B\kern-.05em{\sc i\kern-.025em b}\kern-.08em
    T\kern-.1667em\lower.7ex\hbox{E}\kern-.125emX}}

\begin{document}

\title{A Novel Framework for Honey-X Deception in Zero-Sum Games}
\author{Brendan T. Gould, \IEEEmembership{Student Member, IEEE}, Kyriakos G. Vamvoudakis, \IEEEmembership{Senior Member, IEEE}
\thanks{This work was supported in part by NSF under grant Nos.  CPS-$2227185$ and SATC-$2231651$, and by ARO under grant No. W$911$NF-$24-1-0174$.}
\thanks{Brendan T. Gould is with the {School of Electrical and Computer Engineering, Georgia Institute of Technology, Atlanta, GA 30332 USA (e-mail: bgould6@gatech.edu). }}
\thanks{Kyriakos G. Vamvoudakis is with the {Daniel Guggenheim School of Aerospace Engineering, Georgia Institute of Technology, Atlanta, GA 30332 USA (e-mail: kyriakos@gatech.edu). }}
}

\maketitle

\begin{abstract}
    In this paper, we present a novel, game-theoretic model of deception in two-player, zero-sum games.
    Our framework leverages an information asymmetry: one player (the \emph{deceiver}) has access to accurate payoff information, while the other (the \emph{victim}) observes a modified version of these payoffs due to the deception strategy employed.
    The deceiver's objective is to choose a deception-action pair that optimally exploits the victim's best response to the altered payoffs, subject to a constraint on the deception's magnitude.
    We characterize the optimal deceptive strategy as the solution to a bi-level optimization problem, and we provide both an exact solution and an efficient method for computing a high-quality feasible point.
    Finally, we demonstrate the effectiveness of our approach on numerical examples inspired by honeypot deception.
\end{abstract}

\begin{IEEEkeywords}
Game Theory, Deception, Security
\end{IEEEkeywords}









\section{Introduction}
\label{sec:intro}

In today's information-rich environment, security principles and requirements are continually evolving. To maintain safety, intelligent security strategies are essential in both physical~\cite{tambeSecurityGameTheory2011} and cyber~\cite{pawlickGametheoreticTaxonomySurvey2019, pawlickGameTheoryCyber2021} domains.
Deception plays a central role in the development of such strategies. It is a powerful mechanism employed by both attackers and defenders to gain strategic advantage by obscuring their true intentions, methods, or resources from adversaries~\cite{baoCyberDeceptionTechniques2023}. Game theory, which models strategic interactions through the notions of \emph{agents}, \emph{strategies}, and \emph{payoffs}~\cite{basarDynamicNoncooperativeGame1999}, serves as the predominant framework for capturing the conflicts and dynamics inherent in deception.
A wide range of game-theoretic settings offers fertile ground for analyzing deception. Previous research has explored models based on static~\cite{hespanhaDeceptionNonCooperativeGames2000}, hierarchical~\cite{ganImitativeFollowerDeception2019}, and dynamic~\cite{fotiadisConcurrentRecedingHorizon2023} games.


\paragraph*{Related Work}

One widely adopted modeling framework for deception is the signaling game~\cite{sobel2020signaling}, played between a sender and a receiver. The sender has valuable information, such as the realization of a random variable, and crafts a message to communicate with the receiver. Based on the sender's communication policy, the receiver updates their belief about the state and selects an action. This setup naturally allows for malicious signaling, where the sender may intentionally mislead the receiver, by promoting incorrect hypotheses~\cite{zhangHypothesisTestingGame2018} or inducing a specific false belief~\cite{basarInducementDesiredBehavior2024}. When the receiver performs Bayesian updates, this form of influence is referred to as \emph{Bayesian Persuasion}~\cite{kamenicaBayesianPersuasion2011, bergemannInformationDesignUnified2019, sayinHierarchicalMultistageGaussian2019}.

Extensions of the signaling game framework further explore the interaction between deception and communication policies. For example, \cite{pawlickModelingAnalysisLeaky2019} introduced a ``detector'' mechanism for the receiver, which probabilistically flags deceptive messages, thereby imposing a cost on deception. Other studies have examined sensor manipulation games, in which the receiver relies on sensors to estimate the state, while the sender can manipulate these sensors or selectively reveal them~\cite{hespanhaSensorManipulationGames2021, vamvoudakisDetectionAdversarialEnvironments2014}.

Another prominent framework for modeling deception is the Stackelberg game~\cite{basarDynamicNoncooperativeGame1999}, where a leader commits to a strategy, anticipating the responses of the followers. This structure is well-suited for security applications, where defenders must take into account potential attacks~\cite{tambeSecurityGameTheory2011}. However, if followers strategically conceal their payoffs, it becomes much more challenging for the leader to devise a robust policy~\cite{ganImitativeFollowerDeception2019, nguyenImitativeAttackerDeception2019}.

Alternatively, the leader may engage in deception by deviating from their declared strategy, crafting announcements that induce specific follower beliefs while strategically selecting a true strategy that exploits those beliefs. In~\cite{nguyenWhenCanDefender2022}, the leader is allowed to choose any strategy within an interval surrounding their announced policy. Another model~\cite{kiekintveldGameTheoreticFoundationsStrategic2015, pibilGameTheoreticModel2012} decomposes the leader's action into ``visible'' and ``hidden'' components, modelling ``honeypots'' in cybersecurity. In such settings, the follower may conduct reconnaissance to infer the allocation of security resources, while the leader can covertly introduce decoy targets with no intrinsic value to entrap adversaries. Honey-X techniques use information to induce victims to select actions that benefit the deceiver, luring a cyber-attacker to engage with a fake honeypot rather than a genuine server~\cite{pawlickGametheoreticTaxonomySurvey2019}.

Our model generalizes prior honeypot designs~\cite{kiekintveldGameTheoreticFoundationsStrategic2015, pibilGameTheoreticModel2012}, which restricted deception to adding fake targets. We allow broader modifications, including disguising the value of real targets. Unlike~\cite{nguyenImitativeAttackerDeception2019}, which studied \emph{follower} deception, we focus on the leader. While~\cite{nguyenWhenCanDefender2022} simplified deception to strategy announcements, our approach directly models pay-off manipulation, capturing the strategic reasoning of agents more accurately. By ensuring only ``stealthy'' deceptions are used, we also mitigate the risk of honeypot detection. Furthermore, we emphasize the rationality assumptions on both the deceiver and the victim. Since Honey-X deception aims to influence victim decisions, understanding their response is critical~\cite{nguyenWhenCanDefender2022}. 

\paragraph*{Contributions} The contribution of this paper is multi-fold. We develop a Honey-X deception framework for Stackelberg games with zero-sum payoffs represented as a matrix. The leader, knowing the true payoffs, announces a modified matrix to mislead the follower, subject to stealth and budget constraints. The follower, unaware of the true payoffs, selects an action based on the announced matrix. The deception is designed to induce the follower to choose strategies that benefit the deceiver.  We consider two behavioral models: ``a trusting'' victim who treats the announced payoffs as accurate, and ``a robust'' victim who accounts for the worst-case deception consistent with the announcement. Surprisingly, both models yield identical responses, suggesting that our framework captures a broad spectrum of victim behaviors. These behaviors can be explicitly characterized as the solution set of a linear program (LP). Consequently, designing optimal deception reduces to modifying the coefficients of the program, specifically the constraint \emph{matrix}, to shift this solution to a desirable location.

\paragraph*{Structure} The remainder of the paper is organized as follows. Section 2 introduces the preliminaries and formulates the problem of deception in zero-sum games in a Stackelberg framework. Section 3 derives optimal strategies for both the deceiver and the victim, including exact and approximate methods to calculate the deception strategies. In Section 4, numerical examples illustrate the effectiveness and computational properties of the proposed deception strategies. Section 5 concludes and outlines directions for future research.




\paragraph*{Notation}
\label{subsec:notation}

Subscripts to constants denote their size; $e_{i}$ is the $i$th standard basis vector.
Subscripts to vectors and matrices are indexing operations: for $x \in \mathbb{R}^{n}$ and $A \in \mathbb{R}^{n \times n}$, $x_{i}$ is the $i$th component of $x$, $A_{ij}$ is the $i,j$th entry of $A$, and $A_{i}$ is the $i$th row of $A$.
The element-wise ordering $\le$ compares two vectors. 
Let $\Delta(z)$, $z \in \Z^+$ be the simplex
\begin{equation}
    \label{eq:simplex}
	\Delta(z) = \left\{ p \in \R^z \mid p \ge 0_z, \sum_{i = 1}^z p_i = 1 \right\}.
\end{equation}
For any $x \in \mathbb{R}^{n}$, we use the $\ell$-norm $\lVert x \rVert_{\ell} = \left( \sum_{i=1}^{n} \lvert x_{i} \rvert^{\ell}  \right)^{\frac{1}{\ell}}$.
This norm is extended to linear operators $A \in \mathbb{R}^{m \times n}$ by $\lVert A \rVert_{\ell} = \sup_{\lVert x \rVert_{\ell} = 1} \lVert Ax \rVert_{\ell}$.
In the special cases $\ell = 1$ or $\ell = \infty$ (defined in the limiting sense), we have explicit formulae:
\begin{subequations}
     \label{eq:norms}   
    \begin{align}
    	\lVert A \rVert_{1}      & = \max_{1 \leq j \leq n} \sum_{i=1}^{m} \lvert A_{ij} \rvert = \max_{1 \le j \le n} \left\lVert A^{\top}_{j} \right\rVert_{1}, \label{eq:norms:one_norm} \\
    	\lVert A \rVert_{\infty} & = \max_{1 \leq i \leq m} \sum_{j=1}^{n} \lvert A_{ij} \rvert = \max_{1 \le i \le m} \left\lVert A_{i} \right\rVert_{1}. \label{eq:norms:inf_norm}
    \end{align}
\end{subequations}

\section{Preliminaries and Problem Formulation}
\label{sec:model}

As a basis for our work, we provide a brief review of zero-sum games and a mathematical problem formulation.

\subsection{Zero Sum Games}
\label{subsec:game_setup}

Consider a two-player, zero-sum game, represented by a matrix $G \in \Gspace$.
The \emph{row player} (he) chooses an action $x \in \rowspace$ and the \emph{column player} (she) chooses $y \in \colspace$.
Given these choices, the outcome of the game is given by
\begin{equation}
	\label{eq:game_outcome}
	v_G(x, y) := x^\top G y.
\end{equation}
We interpret this as a payment from the row player to the column player.
Consequently, one player chooses $x$ to minimize $v_G(x, y)$, while the other chooses $y$ to maximize the same quantity. Since each player's payoff depends not only on their own action but also on the opponent's, additional specification is required to describe the solution to this simultaneous optimization problem.
The conventional solution concept for such game-theoretic scenarios is a \emph{Nash equilibrium}~\cite{basarDynamicNoncooperativeGame1999}.

\begin{definition}[Nash Equilibrium]
    For a game $G \in \Gspace$, the tuple $(\rowNE, \colNE) \in \rowspace \times \colspace$ is a \emph{Nash equilibrium} if 
    \begin{subequations}
        \label{eq:nashEQ}
        \begin{align}
            v_G(\rowNE, \colNE) &\le v_G(x, \colNE) \quad \forall x \in \rowspace, \\ 
            v_G(\rowNE, \colNE) &\ge v_G(\rowNE, y) \quad \forall y \in \colspace.
        \end{align}
    \end{subequations}
\end{definition}

Nash equilibria of zero-sum games can be efficiently computed by assuming that both agents act as worst-case optimizers, choosing strategies $\rowNE$ and $\colNE$ such that
\begin{subequations}
    \label{eq:zs_game_minmax}
	\begin{align}
		\rowNE \in \argmin_{x \in \rowspace} \max_{y \in \colspace} v_{G}(x, y), \label{eq:zs_game_minmax:row} \\
		\colNE \in \argmax_{y \in \colspace} \min_{x \in \rowspace} v_{G}(x, y). \label{eq:zs_game_minmax:col}
	\end{align}
\end{subequations}
The computation of the column player's worst-case optimal strategy can be expressed as a LP:
\begin{equation}
	\label{eq:game_as_lp}
	\colNE \in \argmax_{y \in \colspace} v \quad \text{s.t.}\, \begin{bmatrix} -G & 1_{m \times 1} \end{bmatrix} \begin{bmatrix} y \\ v \end{bmatrix} \le 0. \\
\end{equation}
The inequality constraint in~\eqref{eq:game_as_lp} reflects the worst-case nature of the column player's optimization, and ensures that she receives a payoff of at least $v$ regardless of the row player's action.
Using $x$, the row player can select any convex combination of the entries of $Gy$ as an outcome.
Therefore, forcing a lower bound on each of these entries also gives a lower bound on his achievable outcomes.
This idea defines \emph{security value}.

\begin{definition}[Security Strategy and Value~{\cite[Definition 2.3]{basarDynamicNoncooperativeGame1999}}]
    Let $G \in \Gspace$ be a zero-sum game. 
    A strategy $\colSP$ is a column player \emph{security strategy} or \emph{policy} if 
    \begin{equation}
        \label{eq:security_strategy_col}
        \colSV(G) := \min_{x \in \rowspace} x^\top G \colSP \ge \min_{x \in \rowspace} x^\top G y' \quad \forall y' \in \colspace.
    \end{equation}
    We call $\colSV(G)$ the \emph{security value}. 
    The row player's policies $\rowSP$ and value $\rowSV(G)$ are defined respectively.
\end{definition}

The result of the classic minimax theorem~\cite{v.neumannZurTheorieGesellschaftsspiele1928} is that, for any two-player zero-sum game $G$, $\rowSV(G) = \colSV(G)$.
Furthermore, when both agents are playing security policies $\rowSV$ and $\colSV$, the outcome corresponds to these security values and is a Nash equilibrium~\cite{basarDynamicNoncooperativeGame1999}: 
\begin{equation}
    \label{eq:game_value}
    v_G(\rowSP, \colSP) = v_G(\rowNE, \colNE) = \rowSV(G) = \colSV(G).
\end{equation}
This allows us to assign a single value to the game.
With abuse of notation, we  write simply $v_{G}$ to represent $v_{G}(\rowNE, \colNE)$, where $\rowNE$ and $\colNE$ are given by~\eqref{eq:zs_game_minmax}.

%
%

\subsection{Problem Formulation}
\label{subsec:deception}

To enable a \emph{deceiver} to declare information that influences the behavior of a \emph{victim}, we extend the zero-sum game structure from Section~\ref{subsec:game_setup} to a Stackelberg framework. The row player acts as both the leader and the deceiver, possessing full knowledge of the game matrix $G$. This player selects a \emph{deception} $D \in \mathbb{R}^{m \times n}$ to alter the information available to the column player, who observes the perturbed game
\begin{equation} \label{eq:G_pert} \Gpert := G + D, \end{equation}
without knowing the exact values of either $G$ or $D$. The follower's payoffs are still determined by the true matrix $G$. This setup is illustrated in Figure~\ref{fig:deception_timeline}.

\begin{figure}
    \centering
    \includesvg[width=\linewidth]{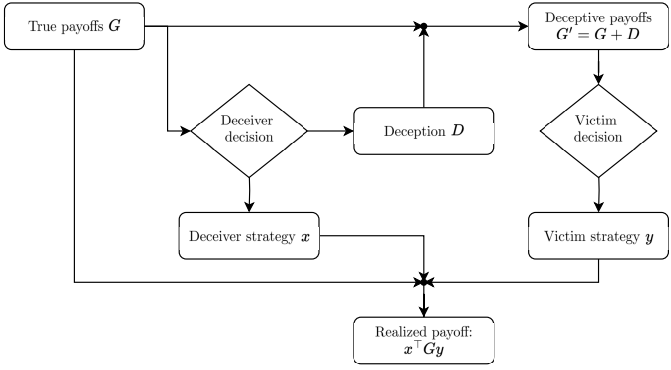}
    \caption{Flowchart depiction of our model of deception. 
    Rectangles represent numeric quantities, and parallelograms identify strategic decisions.
    Arrows indicate causation; the source of the arrow affects the node at its sink.}
    \label{fig:deception_timeline}
\end{figure}

We focus exclusively on \emph{stealthy} deceptions, those that cannot be plausibly distinguished from the truth by the victim. To ensure this, we impose two key restrictions. First, unlike prior work~\cite{nguyenWhenCanDefender2022}, the deceiver may \emph{not} announce a deceptive strategy. Instead, we implicitly assume that announcing $\Gpert$ includes a corresponding security policy for $\Gpert$. If a strategy inconsistent with such a policy were announced, a meta-rational victim could infer that a rational deceiver would not play it, thereby detecting the deception. Crucially, this does not require the deceiver to actually follow the announced strategy, nor does it imply that the victim assumes that they will.

Secondly, we constrain the deceiver's capabilities using a deception budget $\Dmag \in \mathbb{R}^+$. Specifically, we assume

\begin{equation} \label{eq:bounded_deception} \norm{D}_1 \le \Dmag. \end{equation}

Let $\Dspace := { D \in \mathbb{R}^{m \times n} : \lVert D \rVert_{1} \le \Dmag }$ denote the set of admissible deceptions, known to both players. This reflects the victim's limited prior knowledge of the true rewards $G$ and ensures that the deceiver remains stealthy by restricting the magnitude of changes in the announced game $\Gpert$.

From the victim's perspective, it is important to mitigate the effects of misinformation. This leads to computing a security policy that is robust against all admissible stealthy deceptions:

\begin{equation} \colRobust \in \argmax_{y \in \colspace} \min_{x \in \rowspace, D \in \Dspace} v_{\Gpert - D}(x, y). \label{eq:victim_worst_case} \end{equation}

This formulation differs from~\eqref{eq:zs_game_minmax:col} in two ways. First, the objective is written as $v_{\Gpert - D}(x, y)$ instead of $v_G(x, y)$; by~\eqref{eq:G_pert}, these are equivalent, but the former emphasizes that the victim does not know $G$ and cannot compute $v_G(x, y)$ directly. Second, the deceiver now has an additional decision variable—$D$—which increases his influence over the outcome. Despite this added complexity, we show in Theorem~\ref{thm:victim_worst_case} that computing the victim's security policy remains tractable.

The deceiver's problem is the natural dual: how to select the optimal deception-action pair against any possible victim response. This leads to the following minimax formulation:

\begin{equation} (\rowRobust, \deceptionRobust) \in \argmin_{x \in \rowspace, D \in \Dspace} \max_{y \in \colspace} v_G(x, y). \label{eq:deceiver_worst_case} \end{equation}

The solution to~\eqref{eq:deceiver_worst_case} is straightforward. Since the victim could, by chance, select an optimal security policy for the true game $G$, any deception may be rendered ineffective. Thus, the best strategy for the deceiver is simply to play a Nash equilibrium of $G$.
This result is formalized in Proposition~\ref{prop:attacker_worst_case}.

While~\eqref{eq:deceiver_worst_case} identifies the best deception that is \emph{robust} to any victim strategy, it is also valuable to determine the optimal deception for a \emph{specific} victim decision-making model. We consider two such models: ``a trusting'' victim who selects actions according to~\eqref{eq:zs_game_minmax:col} for $\Gpert$, and ``a robust'' victim who solves~\eqref{eq:victim_worst_case}. In Theorem~\ref{thm:victim_worst_case}, we show that both models yield identical behavior. This agreement supports the predictive validity of these models and motivates the following assumption.

\begin{assumption}[Victim Behavior]
    \label{asmp:victim_behavior}
    Let $\Gpert$ be as in~\eqref{eq:G_pert}. 
    We assume that the victim chooses actions 
    \begin{equation}
        y \in \victimStrategies{\Gpert}:= \{ y \in \colspace \mid \Gpert y \ge v_{\Gpert} \}.
    \end{equation}
    Note that by~\eqref{eq:game_as_lp}, $\victimStrategies{\Gpert}$ is exactly the set of column player security policies for $\Gpert$.
\end{assumption}

Under Assumption~\ref{asmp:victim_behavior}, the deceiver can anticipate the victim's response to any chosen deception $D$, and thereby determine their own best action $x$. This prediction involves solving a LP~\eqref{eq:game_as_lp}, making the computation of the optimal deception-action pair a bi-level optimization problem:
\begin{equation}
    (\rowIntelligent, \deceptionIntelligent, \colIntelligent) \in \argmin_{\substack{x \in \rowspace, D \in \Dspace, \\ y \in \colspace}} v_G(x, y) \text{ s.t. } y \in \victimStrategies{G+D}.  \label{eq:blp}
\end{equation}
In Section~\ref{subsec:deceiver_results_exact}, we reduce~\eqref{eq:blp} to a single-level optimization problem solvable via branch-and-bound techniques. However, due to its computational complexity, this approach does not scale well to larger instances. Thus, Section~\ref{subsec:deceiver_results_feas} introduces an efficient method for computing a feasible point of~\eqref{eq:blp}. This algorithm solves a logarithmic number of LPs and yields a deception-action pair that still effectively exploits the victim's behavior.
\begin{remark}
	\label{rmk:optimism}
Although the column player chooses $y$ to \emph{maximize} $v_G(x, y)$, $y$ appears as a \emph{minimizer} in~\eqref{eq:blp}. Under Assumption~\ref{asmp:victim_behavior}, the victim selects a Nash strategy for the perturbed game $\Gpert$, ensuring her own goals are prioritized. Using $y$ as a minimizer reflects the optimistic assumption that from all such strategies, the victim selects the one giving the deceiver the best outcome in the true game $G$. 
This corresponds to a \emph{strong Stackelberg equilibrium}, and is well-established in the literature~\cite{tambeSecurityGameTheory2011, conitzerComputingOptimalStrategy2006, osborneCourseGameTheory1994}. While a pessimistic perspective may seem more aligned with the competitive nature of the model, it is rarely adopted because weak Stackelberg equilibria may not exist. Some work~\cite{guoInducibilityStackelbergEquilibrium2019} has attempted to generalize the pessimistic view by focusing on \emph{inducible} equilibria.
\end{remark}

\begin{remark}
    \label{rmk:bilevel_LP}
Although both the objective and the lower-level problem~\eqref{eq:game_as_lp} in~\eqref{eq:blp} are linear, the overall formulation is not a bi-level LP in the classical sense~\cite{ben-ayedBilevelLinearProgramming1993}. The deception variable $D$ appears in the upper level as a decision variable but enters the lower-level constraints bi-linearly with $y$, resulting in a non-convex dependency between levels. This coupling makes the feasible region of the lower-level problem dependent on upper-level decisions in a complex way. To our knowledge, only one algorithm~\cite{liuValueFunctionBasedSequentialMinimization2023a} attempts to solve such problems, but it yields only approximate solutions and cannot guarantee feasibility within domains such as the probability simplex $\colspace$, rendering it unsuitable for our purposes.
\end{remark}

\section{Deceptive Strategy Derivation}
\label{sec:contributions}

We now derive optimal strategies for both the victim and the deceiver in the deceptive game introduced in Section~\ref{subsec:deception}. In Section~\ref{subsec:victim_results}, we show that a worst-case optimizing victim cannot outperform simply treating the deceptive game $\Gpert$ as the true payoff matrix. Section~\ref{subsec:deceiver_results_exact} then formulates a program to calculate the optimal strategy of the deceiver as defined in~\eqref{eq:blp}. Due to scalability limitations of this exact method, Section~\ref{subsec:deceiver_results_feas} presents a more efficient approach for computing a feasible solution to~\eqref{eq:blp}, by solving a logarithmic number of LPs while effectively leveraging deception.

\subsection{Victim Strategy Derivation}
\label{subsec:victim_results}

In~\eqref{eq:victim_worst_case}, we introduced the problem of how a worst-case optimizing victim should select an action in a deceptive game. We now solve this problem and establish a counterintuitive result: the optimal robust deception strategy is to behave as though the perturbed game $\Gpert$ reflects the true rewards - effectively ignoring the possibility of deception.

\begin{theorem}
	\label{thm:victim_worst_case}
Let $\victimStrategies{\Gpert}$ be defined as in Assumption~\ref{asmp:victim_behavior} and $\victimStrategiesRobust{\Gpert}$ be the solution set of~\eqref{eq:victim_worst_case}.
Then, $\victimStrategiesRobust{\Gpert} = \victimStrategies{\Gpert}$.
\end{theorem}

\begin{proof}
	Note that
	\begin{align*}
		\victimStrategiesRobust{\Gpert} &= \argmax_{y \in \colspace} \min_{x \in \rowspace, D \in \Dspace} v_{G' - D}(x, y)                                             \\
		                       & = \argmax_{y \in \colspace} \min_{x \in \rowspace} x^{\top} G' y - \max_{D \in \Dspace} x^{\top} D y,
	\end{align*}
	since the domains of $x$ and $D$ are independent.
	We simplify the inner maximization by applying the properties of \emph{dual norms}~\cite[Appendix A.1]{boydConvexOptimization2004}:
	\begin{align*}
		\max_{D \in \Dspace} x^\top D y = \max_{\left\lVert \frac{D}{\Dmag} \right\rVert_1 \le 1} \left\langle \frac{D}{\Dmag}, \Dmag x y^\top \right\rangle = \Dmag \lVert xy^\top \rVert_\infty .
	\end{align*}
	Using the explicit form for the $\lVert \cdot \rVert_\infty$ norm~\eqref{eq:norms:inf_norm}, we find
	\begin{equation*}
		\Dmag \lVert x y^\top \rVert_\infty = \Dmag \max_{1 \le i \le m} \sum_{j = 1}^n \vert x_i y_j \rvert = \Dmag \max_{1 \le i \le m} x_i,
	\end{equation*}
	since both $x$ and $y$ are elements of probability simplices.
    
    We now have
	\begin{equation*}
		\victimStrategiesRobust{\Gpert} = \argmax_{y \in \colspace} \min_{x \in \rowspace} x^\top G' y - \Dmag \max_{1 \le i \le m} x_i.
	\end{equation*}
    Consider any solution $x^\star$ and $y^\star$ to the inner and outer optimization problems, respectively. 
    We claim that $x^\star = e_j$, where $j \in \argmin_{1 \le i \le m} (G'y)_i$.
    Assume by contradiction that $x^\star$ does not have this form, then 
    \begin{equation*}
        e_j^\top \Gpert y^\star \le x^{\star\top} \Gpert y^\star \text{ and } - \Delta \max_{1 \le i \le m} (e_j)_i < -\Delta \max_{1 \le i \le m} x^\star_i.
    \end{equation*}
    Thus, $x^\star$ is \emph{not} a minimizer of the inner problem, the contradiction. 
    But then, $-\Delta \max_{1 \le i \le m} x^\star_i = -\Delta$, and 
    \begin{equation*}
        \victimStrategiesRobust{\Gpert} = \argmax_{y \in \colspace} \min_{x \in \rowspace} x^\top \Gpert y - \Delta = \victimStrategies{\Gpert},
    \end{equation*}
    since the inner objectives differ only by a constant. 
\end{proof}

\begin{remark}
\label{rmk:victim_worst_case_ex}
Theorem~\ref{thm:victim_worst_case} establishes that the behavior of a trusting victim—who assumes $\Gpert$ reflects the true payoff matrix—is indistinguishable from that of a robust victim, who anticipates deception and selects actions to mitigate the worst-case scenario. Consequently, no choice of $\Gpert$ provides the victim with any actionable insight into the actual payoffs $G$. Any deviation from a security strategy optimized for $\Gpert$ would degrade the victim's deception-aware security value. This is advantageous in defensive settings, as it enables the deceiver to strategically design $\Gpert$ to influence the victim's behavior without compromising private information about $G$.
\end{remark}

\subsection{Deceiver Strategy Derivation}
\label{subsec:deceiver_results_exact}
We now examine the deceptive game introduced in Section~\ref{sec:model} from the perspective of the deceiver. As a natural counterpart to the victim's analysis in Section~\ref{subsec:victim_results}, we begin by identifying the optimal deceptive strategy under the assumption that the victim responds in a worst-case manner.

\begin{proposition}
\label{prop:attacker_worst_case}
Consider any game $G \in \Gspace$. Let $D \in \Dspace$, and let $\rowNE$ be a Nash equilibrium strategy for the row player in $G$. Then, the pair $(D, \rowNE)$ solves~\eqref{eq:deceiver_worst_case}. As in standard zero-sum games, this strategy has a security value of $v_G$.
\end{proposition}

\begin{proof}
Note that
\begin{equation}
\min_{x \in \rowspace, D \in \Dspace} \max_{y \in \colspace} v_G(x, y) = \min_{x \in \rowspace} \max_{y \in \colspace} v_G(x, y),
\end{equation}
since $D$ does not affect the payoff function $v_G(x, y)$, which depends solely on the true game $G$. Therefore, the optimal strategy is simply a Nash equilibrium of $G$, and any choice of $D \in \Dspace$ does not influence the outcome.
\end{proof}

\begin{remark}
\label{rmk:attacker_worst_case_ex}    
Although the victim lacks knowledge of the true payoffs in $G$, it remains possible that their chosen strategy $y$ coincides with a Nash equilibrium. Without assuming a specific decision-making model, the deceiver cannot eliminate this possibility. Consequently, a worst-case optimizing deceiver will always adopt a Nash strategy for $G$, effectively disregarding their deceptive capabilities.
\end{remark}
To go beyond the simplistic strategy described in Proposition~\ref{prop:attacker_worst_case}, we adopt a specific decision-making model for the victim, as in Assumption~\ref{asmp:victim_behavior}. According to Theorem~\ref{thm:victim_worst_case}, this model captures both trusting and robust behaviors.

The corresponding optimization problem is posed in~\eqref{eq:blp}, which we now reformulate as a single-level program that can be addressed using branch-and-bound techniques. Specifically, consider the following nonlinear, non-convex program:
\begin{subequations}
	\label{eq:blp:bilinear}
	\begin{align}
		\min_{\substack{x \in \rowspace, D \in \Dspace, \\ y \in \colspace, \omega \in \rowspace, \vPerceived \in \R}} & v_G(x, y), \\
		 \text{s.t. } & (G+D)y \ge \vPerceived,                   \\
		 & (G+D)^\top \victimDualAction \le \vPerceived.
	\end{align}
\end{subequations}


\begin{theorem} \label{thm:deceiver_exact}  Let $(x^\star, D^\star, y^\star, \victimDualAction^\star, \vPerceived^\star)$ be an optimal point of~\eqref{eq:blp:bilinear}. 
Then, the tuple $(x^\star, D^\star, y^\star)$ is also an optimal solution to the bi-level program~\eqref{eq:blp}. \end{theorem}

\begin{proof} We demonstrate the equivalence of the two formulations by reducing~\eqref{eq:blp} to~\eqref{eq:blp:bilinear}.
First, recall that by the definition of $\victimStrategies{\Gpert}$, the bi-level program~\eqref{eq:blp} can be expanded as:

\begin{subequations} 
\label{eq:blp2}
\begin{alignat}{2}
    \min_{\substack{x \in \rowspace, D \in \Dspace, \\ y \in \colspace, \vPerceived \in \R}} & v_G(x, y), \label{eq:blp2_ul_obj} \\
    \text{s.t. } \quad & (y, \vPerceived) \in \argmax_{y \in \colspace, \vPerceived \in \R} \vPerceived, \label{eq:blp2_ll_obj} \\
    & \text{s.t. } \begin{bmatrix} -(G + D) & 1_{m \times 1} \end{bmatrix} \begin{bmatrix} y \\ \vPerceived \end{bmatrix} \le 0. \label{eq:blp2_ll_cons}
\end{alignat}
\end{subequations}

Next, we apply the principle of strong duality.  
The dual of the lower-level LP is:

\begin{subequations}
\label{eq:dual_deceiver_lp}
\begin{align}
    (\victimDualAction, \vDualPerceived) \in & \argmin_{\victimDualAction \in \rowspace, \vDualPerceived \in \R} \vDualPerceived \\
    & \text{s.t. } (G+D)^\top \victimDualAction \le \vDualPerceived.
\end{align}
\end{subequations}

Since both the primal and dual programs are linear, strong duality holds~\cite{boydConvexOptimization2004}.  
Therefore, a pair of feasible solutions $(y, \vPerceived)$ to~\eqref{eq:blp2_ll_obj}–\eqref{eq:blp2_ll_cons} and $(\victimDualAction, \vDualPerceived)$ to~\eqref{eq:dual_deceiver_lp} are optimal if and only if $\vPerceived = \vDualPerceived$.  
This equivalence allows us to reduce~\eqref{eq:blp2} to the single-level formulation~\eqref{eq:blp:bilinear}, completing the proof.
\end{proof}

\begin{remark}
\label{rmk:deceiver_exact_ex}    

Theorem~\ref{thm:deceiver_exact} provides an explicit reduction of the bilevel program - used to compute the optimal deceptive strategy against a victim whose behavior satisfies Assumption~\ref{asmp:victim_behavior} - to a single-level.
While established optimization software can solve such problems,~\eqref{eq:blp:bilinear} remains computationally challenging.
Specifically, the bilinear and nonconvex nature of the constraints requires the use of branch-and-bound procedures to obtain a global solution~\cite{mccormickComputabilityGlobalSolutions1976}.
Due to branching, the complexity of these methods grows exponentially with the number of decision variables, rendering~\eqref{eq:blp:bilinear} impractical for real-world scenarios involving large action spaces for both the leader and the follower.
In practice, we solve~\eqref{eq:blp:bilinear} for smaller instances using the Gurobi optimization package~\cite{gurobi}.
In Section~\ref{sec:examples}, we apply~\eqref{eq:blp:bilinear} to compute optimal deception strategies for example games and report the corresponding computation times.

\end{remark}

\subsection{Feasible Deception Derivation}
\label{subsec:deceiver_results_feas}

In order to combat the computational inefficiencies of branch-and-bound methods, we also present a scalable algorithm to identify a feasible point of~\eqref{eq:blp}.
Empirically, this feasible point also improves the deceiver's utility compared to the base zero-sum case, which we demonstrate in Section~\ref{sec:examples}.

We introduce additional notation to clarify the intuition behind this algorithm. 
First, fix $v \in \R$ and define the \emph{victim sub-rational solution set} 
\begin{equation}
    \label{eq:victimSubrationalBehavior}
    \victimStrategiesSubrational{\Gpert}{v} := \{ y \in \colspace \mid \Gpert y \ge v \}.
\end{equation}
Note that $\victimStrategiesSubrational{\Gpert}{v}$ generalizes $\victimStrategies{\Gpert}$ from Assumption~\ref{asmp:victim_behavior} in the sense that $\victimStrategiesSubrational{\Gpert}{v_\Gpert} = \victimStrategies{\Gpert}$.
For $v < v_\Gpert$, we have $\victimStrategies{\Gpert} \subseteq \victimStrategiesSubrational{\Gpert}{v}$.
In context, this relaxes the rationality requirement on the victim; instead of choosing an exact best response, she may make any choice guaranteeing herself at least $v$ utility, extending the optimism described in Remark~\ref{rmk:optimism}.
We make this relaxation because this simplified model of victim behavior makes computing an optimal deception much easier.
The lower level problem of~\eqref{eq:blp} is relaxed from a LP to just a system of linear constraints, reducing the overall program to a single level. 
We emphasize that this sub-rationality is purely an intermediate device -- it is used as a tool to identify a feasible point of~\eqref{eq:blp}, where the victim is fully rational. 

\begin{remark}
    \label{rmk:blp_to_lp}

    The sub-rational solution set~\eqref{eq:victimSubrationalBehavior} is used to relax~\eqref{eq:blp} by replacing $\victimStrategies{\Gpert}$ with $\victimStrategiesSubrational{\Gpert}{v}$.
    In~\eqref{eq:blp2}, this corresponds to fixing a particular $\vPerceived = v$. 
    Since this forces a constant lower-level objective, the resulting problem can then be simplified to a single level. 
    Furthermore, the remaining bi-linear terms ($x^\top G y$ in~\eqref{eq:blp2_ul_obj} and $Dy$ in~\eqref{eq:blp2_ll_cons}) can be removed through additional simplification, making the problem fully linear. 
    First, notice that for any $G \in \Gspace$, $x \in \rowspace$, and $y \in \colspace$, $\min_{1 \le i \le m} e_i^\top Gy \le x^\top G y$ by convexity. 
    Therefore, an equivalent approach is to solve the $m$ problems with linear objective $e_i^\top G y$ and simply keep the best solution. 
    Next, consider any $D \in \Dspace$ and $y \in \colspace$; by~\eqref{eq:simplex} and~\eqref{eq:norms:one_norm}, $\norm{Dy}_1 \le \Delta$. 
    Thus, consider any $d \in \R^m$ with $\norm{d}_1 \le \Delta$ and parameterize $D = \begin{bmatrix} d & d & \cdots & d \end{bmatrix}$.
    For all $y \in \colspace$, we have $Dy = d$, which means that we can directly use $d$ as a decision variable, eliminating the product $Dy$.
    We later exploit these simplifications in Algorithm~\ref{alg:deceiver_feas} to compute relaxations that iteratively approach the feasible set of~\eqref{eq:blp}.
\end{remark}



Note that if $v > v_\Gpert$, then $\victimStrategiesSubrational{\Gpert}{v} = \emptyset$ (lest $v$ be a larger objective in~\eqref{eq:game_as_lp}).
We define the set of \emph{inducible values} for a game $G$ as those that have non-empty sub-rational solution sets under some stealthy deception: 
\begin{equation}
    \label{eq:v_inducible}
    \Vind{G} := \{ v \in \R \mid \exists~D \in \Dspace \text{ s.t. } \victimStrategiesSubrational{G+D}{v} \ne \emptyset \}.
\end{equation}
We shall show that $\Vind{G}$ lets us identify particular parameters causing the easy-to-calculate behaviors of~\eqref{eq:victimSubrationalBehavior} to coincide with the fully-rational behaviors of Assumption~\ref{asmp:victim_behavior}. 
Our first lemma towards this result is that $\Vind{G}$ is downward closed.
\begin{lemma}
	\label{lemma:downward_closed}
    Let $G \in \Gspace$ be a zero-sum game, and define $\Vind{G}$ as in~\eqref{eq:v_inducible}.  
	Then, $\Vind{G}$ is non-empty and for any $v \in \Vind{G}$, if $v' < v$, then $v' \in \Vind{G}$.
\end{lemma}

\begin{proof}
	Let $\vmin := \min_{1 \le i \le m, 1 \le j \le n} G_{ij} - \Delta$.
    By~\eqref{eq:bounded_deception} and \eqref{eq:norms:one_norm}, $\min_{1 \le i \le n} (( G + D )y)_{i} \ge \vmin$ for any $D \in \Dspace$ and $y \in \colspace$.
	Thus, $\victimStrategiesSubrational{G+D}{\vmin} = \colspace$ and $\vmin \in \Vind{G}$.

	Consider any $v\in \Vind{G}$ and after considering ~\eqref{eq:victimSubrationalBehavior} and~\eqref{eq:v_inducible}, there exists a tuple $(D, y)$ with $\min_{1 \le i \le n} ((G + D)y)_{i} \ge v > v'$.
    Thus, $\victimStrategiesSubrational{G+D}{v} \subseteq \victimStrategiesSubrational{G+D}{v'}$.
    Therefore, $\victimStrategiesSubrational{G+D}{v'}$ is non-empty, and so $v' \in \Vind{G}$.
\end{proof}

Using Lemma~\ref{lemma:downward_closed} and a compactness argument, we now show that $\Vind{G}$ attains a maximum.

\begin{lemma}
	\label{lemma:inducible_maximum}
    Let $G \in \Gspace$ be a zero-sum game, and let $\Vind{G}$ as in~\eqref{eq:v_inducible}.  
    Further define $\vIndMax := \sup \Vind{G}$. 
    Then, $\vIndMax \in \Vind{G}$.
\end{lemma}

\begin{proof}
	First, we shall show that $\Vind{G}$ is bounded above.
	Let $\vmax := \max_{1 \le i \le m, 1 \le j \le n} G_{ij} + \Delta + 1$.
	By~\eqref{eq:bounded_deception}, \eqref{eq:norms:one_norm} and the fact that $y \in \colspace$, for any $D \in \Dspace$ we must have $\min_{1 \le i \le n} ((G+D) y)_{i} < \vmax$.
    But then, $\victimStrategiesSubrational{G+D}{\vmax} = \emptyset$ and $\vmax \not \in \Vind{G}$.
    Thus, $\vmax$ upper bounds $\Vind{G}$ by Lemma~\ref{lemma:downward_closed}, so $\vIndMax < \infty$.

	Let $\{ v^{i} \}_{i=1}^{\infty}$ be any sequence approaching $\vIndMax$ from below.
	Note that each $v^{i} \in \Vind{G}$, lest $v^{i} < \vIndMax$ be an upper bound for $\Vind{G}$ by the contrapositive of Lemma~\ref{lemma:downward_closed}.
	By~\eqref{eq:victimSubrationalBehavior} and~\eqref{eq:v_inducible}, there exist $\{ D^{i} \}_{i=1}^{\infty}$ and $\{ y^{i} \}_{i=1}^{\infty}$ such that $D^{i} \in \Dspace$, $y^{i} \in \colspace$, and $\min_{1 \le i \le n} ((G + D^{i})y^{i})_{i} \ge v^{i}$ for all $i$.
	Since both $\Dspace$ and $\colspace$ are bounded, by the Bolzano-Weierstrass theorem, there exist convergent subsequences $\{ D^{i^{k}} \}_{k=1}^{\infty}$ and $\{ y^{i^{k}} \}_{k=1}^{\infty}$; call their limits $\Dlim$ and $\ylim$, respectively.
    By closure, $\Dlim \in \Dspace$ and $\ylim \in \colspace$.
	Then, by properties of limits
	\begin{equation*}
		(G+\Dlim)\ylim = \lim_{k \to \infty} (G + D^{i^{k}})y^{i^{k}} \ge \lim_{k \to \infty} v^{i^{k}} = \vIndMax.
	\end{equation*}
    Thus, $\ylim \in \victimStrategiesSubrational{G+\Dlim}{\vIndMax}$, so $\vIndMax \in \Vind{G}$ as claimed.
\end{proof}

Finally, we derive the key property of $\vIndMax$: any ``sub-rational'' strategy achieving at least $\vIndMax$ utility in $\Gpert$ is fully rational for the deceptive game.

\begin{lemma}
	\label{lemma:rational_point}
    Let $G \in \Gspace$ be a zero-sum game, define $\Vind{G}$ as in~\eqref{eq:v_inducible}, and let $\vIndMax = \sup \Vind{{G}}$.
    Then, for any $D \in \Dspace$ and $y \in \colspace$, $y \in \victimStrategiesSubrational{G+D}{\vIndMax} \implies y \in \victimStrategies{G+D}$.
\end{lemma}

\begin{proof}
	Assume by contradiction that $y \not \in \victimStrategies{G+D}$.
	Then, by~\eqref{eq:game_as_lp}, there exists a $y'$ with $v' := \min_{1 \le i \le m} ((G+D) y')_{i} > \vIndMax$.
    But this implies that $y' \in \victimStrategiesSubrational{G+D}{v'}$, and so $v' \in \Vind{G}$.
	This contradicts the fact that $\vIndMax$ is the maximum element of $\Vind{G}$, completing the proof.
\end{proof}

These Lemmas establish the framework needed for Algorithm~\ref{alg:deceiver_feas}.
Our main goal is to identify the $\vIndMax$ guaranteed to exist by Lemma~\ref{lemma:inducible_maximum}, which we accomplish via a binary search.
The techniques of Remark~\ref{rmk:blp_to_lp} are used to solve the problem on Line~\ref{alg:deceiver_feas:lp} as a LP.

\begin{algorithm}[tb]
	\caption{Feasible Deceptive Strategy}
	\textbf{Input:} Zero-sum game $G \in \Gspace$, deception magnitude bound $\Delta \in \R^+$, binary search tolerance $\binSearchTol \in \R^+$
	\label{alg:deceiver_feas}
	\begin{algorithmic}[1]
		\State $\vbest \gets \infty$
		\State $\vmin \gets \min_{1 \le i \le m, 1 \le j \le n} G_{ij} - \Delta$
		\State $\vmax \gets \max_{1 \le i \le m, 1 \le j \le n} G_{ij} + \Delta$
		\For{$i \in \{1, 2, \ldots, m\}$}
		\While{$\vmax - \vmin > \binSearchTol$}
		\State $\vcurr \gets \frac{ \vmax + \vmin }{2}$
		\If{$\vcurr \in \Vind{G}$}
		\State $\Dcurr, \colCurr \gets \argmin\limits_{D \in \Dspace, y \in \colspace} e_i^\top G y$ s.t. $y \in \victimStrategiesSubrational{G+D}{\vcurr}$ \label{alg:deceiver_feas:lp}
		\State $\vmin \gets \vcurr$
		\Else
		\State $\vmax \gets \vcurr$
		\EndIf
		\EndWhile
		\If{$e_{i}^\top G \colCurr < \vbest$} \Comment{Update best solution}
		\State $\vbest \gets e_{i}^\top G \colCurr$
		\State $\rowIntelligentFeas \gets e_{i}$
        \State $\deceptionIntelligentFeas \gets \Dcurr$
        \State $\colIntelligentFeas \gets \colCurr$
		\EndIf
		\EndFor
		\Return $\vbest, \rowIntelligentFeas, \deceptionIntelligentFeas, \colIntelligentFeas$
	\end{algorithmic}
\end{algorithm}

\begin{theorem}
\label{thm:deceiver_feas}
    For any zero-sum game $G \in \Gspace$, let $\{ \binSearchTol^i \}_{i=1}^\infty$ be any sequence approaching 0 from above. 
    Define the corresponding sequence of strategies $\binSearchOutSeq = \{ (\rowIntelligentFeas^i, \deceptionIntelligentFeas^i, \colIntelligentFeas^i) \}_{i=1}^{\infty}$ by the outputs of Algorithm~\ref{alg:deceiver_feas} for each binary search tolerance $\binSearchTol^i$. 
    Then, $\binSearchOutSeq$ converges to a feasible point of~\eqref{eq:blp}.
\end{theorem}

\begin{proof}
    Let $\vIndMax = \sup \Vind{G}$. 
    Consider~\eqref{eq:blp} with the additional constraint that $v_{G+D} = \vIndMax$; by Lemma~\ref{lemma:inducible_maximum} the restricted problem is still feasible, and by~\eqref{eq:game_as_lp}, can be written (as with~\eqref{eq:blp2})
	\begin{subequations} \label{eq:blp_fixed_sup}
	\begin{align}
		\min_{\substack{x \in \rowspace, D \in \Dspace,                                                                                                   \\ y \in \colspace}} & v_G(x, y),  \label{eq:blp_fixed_sup:obj}                                                                                                       \\
		\text{s.t. } & \begin{bmatrix} -(G + D) & 1_{m \times 1} \end{bmatrix} \begin{bmatrix} y \\ \vIndMax \end{bmatrix} \le 0. \label{eq:blp_fixed_sup:rationality} \\
		             & v_{G+D} = \vIndMax \label{eq:blp_fixed_sup:deception_cons}
	\end{align}
	\end{subequations}
    By Lemma~\ref{lemma:rational_point}, \eqref{eq:blp_fixed_sup:deception_cons} is subsumed by~\eqref{eq:blp_fixed_sup:rationality} and can be ignored. 

    Using a binary search, Algorithm~\ref{alg:deceiver_feas} finds
	\begin{subequations} \label{eq:blp_bin_tol}
	\begin{align}
		(\rowIntelligentFeas^i, \deceptionIntelligentFeas^i, \colIntelligentFeas^i) \in \argmin_{\substack{x \in \rowspace, d \in \R^m,                                                                                                   \\ y \in \colspace}} & v_G(x, y),  \label{eq:blp_bin_tol:obj}                                                                                                       \\
    \text{s.t. } & Gy + d \ge (\vIndMax - \theta^i)_m, \label{eq:blp_bin_tol:rationality} \\ 
    & \norm{d}_1 \le \Delta.
	\end{align}
	\end{subequations}
    for some $\theta^i \in [0, \binSearchTol^i]$. 
    Let $C(\theta)$ be the correspondence of the set of feasible points in~\eqref{eq:blp_bin_tol} for each $\theta^i$; $C(\theta)$ is defined by linear constraints and continuous by standard arguments. 
    Furthermore, it is the product of closed subsets of compact domains, and so compact-valued. 
    Finally, since~\eqref{eq:blp_bin_tol} is a relaxation of~\eqref{eq:blp_fixed_sup}, which is feasible as above, $C(\theta)$ is non-empty valued. 
    Therefore, Berge's maximum theorem implies that the set of optimizers of~\eqref{eq:blp_bin_tol} is upper hemi-continuous. 
    Since $\{\binSearchTol^i\}_{i=1}^\infty \to 0$, also $\{ \theta^i \}_{i=1}^\infty \to 0$, and so $\{ (\rowIntelligentFeas^i, \deceptionIntelligentFeas^i, \colIntelligentFeas^i) \}_{i=1}^\infty$ converges to the solution set of~\eqref{eq:blp_fixed_sup}. 
    As above, solutions to~\eqref{eq:blp_fixed_sup} are feasible points of~\eqref{eq:blp}, completing the proof. 
\end{proof}

\begin{remark}
Using the sub-rationality tool~\eqref{eq:victimSubrationalBehavior}, Theorem~\ref{thm:deceiver_feas} provides an alternative way to compute a well-performing deception that is more efficient than solving~\eqref{eq:blp:bilinear}. 
As $\binSearchTol \to 0$, the ``sub-rational'' solution point approaches a point where the victim behaves rationally. 
Therefore, in practice, Algorithm~\ref{alg:deceiver_feas} may be run once with a suitably small $\delta$ to compute a deception that treats the victim as effectively rational. 
Since a binary search is used, $\delta$ may be made very small without excessively increasing the computation required. 
Alternatively, once a deception is computed with a known search tolerance, a post-optimal computation can robustify the result to account for the actions of a rational victim. 
\end{remark}

\begin{proposition}
    \label{prop:deception_bin_search_robustify}
    Consider a zero-sum game $G \in \Gspace$, and pick a binary search tolerance $\binSearchTol > 0$. 
    Let $(\vbest, \rowIntelligentFeas, \deceptionIntelligentFeas, \colIntelligentFeas)$ be given by Algorithm~\ref{alg:deceiver_feas}. 
    When playing $(\rowIntelligentFeas, \deceptionIntelligentFeas)$, the worst-case outcome for the deceiver is upper bounded by 
    \begin{equation}
        \max_{y \in \colspace} \rowIntelligentFeas^\top G y \text{ s.t. } y \in \victimStrategiesSubrational{G+\deceptionIntelligentFeas}{\vbest}, 
    \end{equation}
    which can be computed by solving a LP. 
\end{proposition}

\begin{proof}
    By Assumption~\ref{asmp:victim_behavior}, $y \in \victimStrategies{G + \deceptionIntelligentFeas}$. 
    Since $\colIntelligentFeas$ gives the victim a perceived payoff of at least $\vbest$ for $G + \deceptionIntelligentFeas$, $\vbest \le v_{G + \deceptionIntelligentFeas}$, and so $\victimStrategies{G + \deceptionIntelligentFeas} \subseteq \victimStrategiesSubrational{G + \deceptionIntelligentFeas}{v_{G + \deceptionIntelligentFeas}}$. 
    Therefore, for any $y$ chosen by the victim, 
    \begin{align*}
        \rowIntelligentFeas^\top G y & \le \max_{y \in \colspace} \rowIntelligentFeas^\top G y \text{ s.t. } y \in \victimStrategies{G+\deceptionIntelligentFeas} \\
        & \le \max_{y \in \colspace} \rowIntelligentFeas^\top G y \text{ s.t. } y \in \victimStrategiesSubrational{G+\deceptionIntelligentFeas}{\vbest},
    \end{align*}
    as claimed. 
\end{proof}
\begin{remark}
As a compliment to Theorem~\ref{thm:deceiver_feas}, Proposition~\ref{prop:deception_bin_search_robustify} gives a guaranteed bound on the worst-case deceptive performance of the result computed by Algorithm~\ref{alg:deceiver_feas} for \emph{any} binary search tolerance $\binSearchTol$.
Algorithm~\ref{alg:deceiver_feas} guarantees that $v_{G + \deceptionIntelligentFeas} - \binSearchTol < \vbest < v_{G + \deceptionIntelligentFeas}$, so as $\binSearchTol \to 0$, the conservatism incurred by over-approximating $\victimStrategies{G+\deceptionIntelligentFeas}$ with $\victimStrategiesSubrational{G+\deceptionIntelligentFeas}{\vbest}$ decreases.
Simultaneously, the optimization subroutines in Algorithm~\ref{alg:deceiver_feas} incentivize selecting a $\deceptionIntelligentFeas$ that benefits the deceiver, an advantage over constructing a similar bound for a randomly selected deception. 
In the next section, we examine the performance of $\deceptionIntelligentFeas$ and compare it to the exact deception of Theorem~\ref{thm:deceiver_exact}. 
\end{remark}

\section{Examples}
\label{sec:examples}

We now present numerical examples to complement and illustrate the theoretical results of Section~\ref{sec:contributions}.
Our simulations are implemented in Python and use the optimization package Gurobi~\cite{gurobi} via its official Python interface.
The source code for these examples is available at \texttt{https://tinyurl.com/77cerdsr}; a Gurobi license is required to run some of the larger optimization problems.
All examples were executed on an Arch Linux system equipped with an i9-14000K CPU and 128~GB of RAM.

We begin by comparing the performance of the deception strategies from Theorem~\ref{thm:deceiver_exact}, Theorem~\ref{thm:deceiver_feas}, and Proposition~\ref{prop:deception_bin_search_robustify}.
To approximate the average-case behavior of each method, we use \texttt{numpy}'s uniform sampling to randomly generate a collection of sample games.
For a range of deception budgets $\Delta$, we plot the mean and standard deviation of each method's performance in Figure~\ref{fig:improvement_vs_budget}.
Performance is measured by the difference between the honest value of the game (as in~\eqref{eq:game_value}) and the outcome achieved when the deceiver applies the given strategy and the victim responds according to Assumption~\ref{asmp:victim_behavior}.
We refer to this difference as the method's \emph{improvement}; a larger improvement indicates a more effective deception, while an honest deceiver would always have zero improvement.

\begin{figure}
    \centering
    \includesvg[width=\linewidth]{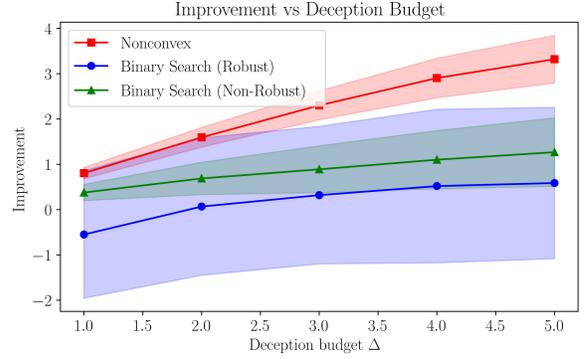}
    \caption{Depiction of the improvement in deceiver utility caused by various forms of deception. 
    $N=100$ games were uniformly sampled from $\Gspace$ with $m=5$ and $n=5$, and our three proposed deceptive strategies were used to compute a deception for each. 
    For deceptions computed via binary search, a numerical tolerance of $\binSearchTol=10^{-3}$ was used.
    The mean and standard deviation of the improvements caused by the computed deceptions are shown.}
    \label{fig:improvement_vs_budget}
\end{figure}

In Figure~\ref{fig:improvement_vs_budget}, we observe that all methods achieve greater improvement as the deception budget $\Delta$ increases, which aligns with intuition.
Moreover, the “non-convex” deception strategy—computed by solving~\eqref{eq:blp:bilinear}—consistently outperforms the simplified binary search methods in terms of improvement.
It is important to note that the “robust” binary search method provides a guaranteed lower bound on the improvement obtained from a selected deception.
In contrast, the “non-robust” method benefits from numerical optimism due to tolerance settings in the binary search, resulting in slightly higher—but less reliable—performance.
Figure~\ref{fig:improvement_vs_budget} suggests that the procedure described in Proposition~\ref{prop:deception_bin_search_robustify}, which removes this optimism, is particularly valuable when the deception budget is small.

The additional improvement achieved by the ``non-convex'' method in Figure~\ref{fig:improvement_vs_budget} comes at a significant computational cost, as illustrated in Figure~\ref{fig:computation_time_vs_size}.
In real-world security scenarios, both defenders and attackers may face large action spaces, making computational efficiency a critical concern.
To assess scalability, we measure the time required to compute a deception strategy for randomly generated games of increasing size, comparing the performance of our two primary approaches.

\begin{figure}
    \centering
    \includesvg[width=\linewidth]{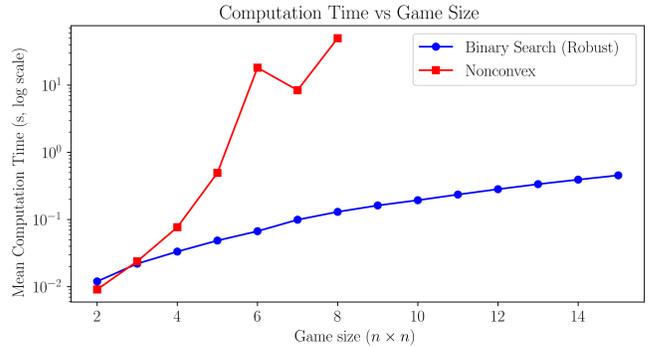}
    \caption{Demonstration of the scaling properties of the time needed to compute a deception for both optimal and binary search deceptive strategies. 
    For each game size $n$, $N = 10$ games were sampled uniformly from $\R^{n \times n}$ and both methods were used to compute a deception for every game with a budget of $\Dmag = 3$. 
    The computation time for both methods were measured, and the mean for each method is displayed here. }
    \label{fig:computation_time_vs_size}
\end{figure}

Figure~\ref{fig:computation_time_vs_size} illustrates that the complexity of solving~\eqref{eq:blp:bilinear} increases rapidly with game size, rendering it impractical in certain settings.
In such cases, the approximate deception strategy computed via Algorithm~\ref{alg:deceiver_feas} offers a viable alternative.

Finally, we examine the impact of the binary search tolerance $\binSearchTol$ on the performance of the deception computed by Algorithm~\ref{alg:deceiver_feas}.
Using a random sample of games of fixed size, we evaluate the improvement achieved by deceptions computed under varying binary search tolerances.
Figure~\ref{fig:improvement_vs_tol} displays both the optimistic improvements reported directly by Algorithm~\ref{alg:deceiver_feas} and the guaranteed bounds provided by Proposition~\ref{prop:deception_bin_search_robustify}.

\begin{figure}
    \centering
    \includesvg[width=\linewidth]{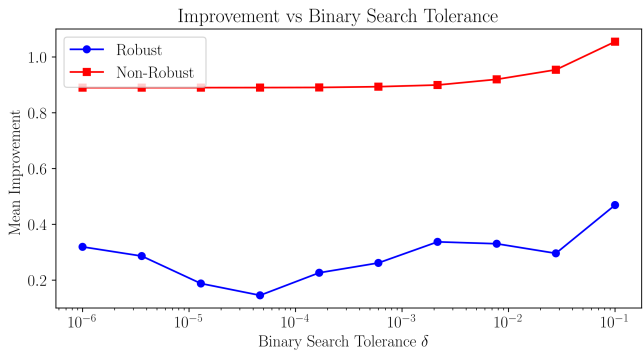}
    \caption{Improvement caused by binary search deception, differentiated between the expected improvement of the non-robust method and the guaranteed improvement of the robust method.
    $N=100$ games were sampled uniformly from $\Gspace$ with $m = 5$ and $n = 5$, and the binary search method was used to compute a deception with budget $\Dmag = 3$ for a range of binary search tolerances $\delta$.
    For each game and tolerance, the expected improvement of the binary search deception and the guaranteed improvement after enforcing robustness were both recorded, and the mean values are displayed here.}
    
    \label{fig:improvement_vs_tol}
\end{figure}

We observe that the robust method's guaranteed improvement is consistently lower than that of the non-robust method, due to the optimistic assumptions embedded in the binary search tolerance $\delta$.
As $\delta$ increases, these assumptions cause the expected improvement of the non-robust method to grow monotonically.
Interestingly, we do not observe a similar trend in the guaranteed improvement as $\delta$ varies.
This is because the bound provided by Proposition~\ref{prop:deception_bin_search_robustify} is influenced by both conservative estimation and the uncertain relationship between the victim's strategy selection for $\Gpert$ and the resulting payoff to the deceiver.

\section{Conclusion}
\label{sec:conclusion}

This work introduced a novel framework for Honey-X deception grounded in Stackelberg game theory.
We analyzed optimal strategies for both the deceiver and the victim, accounting for varying levels of rationality in the victim's response to deception.
In addition, we developed and presented multiple computational approaches that a deceiver can employ to exploit deception, ranging from theoretically optimal formulations to efficiently computable approximations.

For future work we plan to extend the framework to dynamic or repeated interactions, where deception evolves over time and players update beliefs.

\bibliographystyle{ieeetr}
\bibliography{gould_ref}

\end{document}